\documentclass[12pt]{article}
\usepackage{amsmath,amsfonts,amssymb,amsthm, epsfig,
epstopdf, titling, url, array}
\usepackage[utf8]{inputenc}
\usepackage{mathtools}
\usepackage[english]{babel}
\usepackage{hyperref}
\newtheorem{thm}{Theorem}[section]

\newtheorem{cor}{Corollary}[section]
\newtheorem{prop}{Proposition}[section]
\numberwithin{equation}{section}

\def\C{\mathbb C}
\def\D{\mathbb D}

\def\R{I\!\!R}
\def\H{I\!\!H}
\def\C{I\!\!\!\!C}

\def\D{I\!\!D}
\title{Resolvent kernels and time dependent Schr\" odinger equations under magnetic field on the hyperbolic half plane $\H$ and the Morse potential on the real line $\R$ }
\author{Mohamed Vall Ould Moustapha}

\begin{document}
\maketitle
\begin{abstract}

This paper deals with exact  formulas for the resolvent kernels and exact solutions of time dependent  Schr\"odinger equations under a uniform magnetic field on the hyperbolic half plane $\H$, and under a diatomic molecular Morse potential on the real line $\R$. 
\end{abstract}
 
\section{Introduction}
 
In this paper we give an explicit formulas for the Schwartz integral kernels $\widetilde{G}_k(\lambda, z,  z')$, ${\cal G}^{\lambda}_k(\mu, X,  X')$ , $\widetilde{H}_k(\lambda, z,  z')$ and ${\cal H}^{\lambda}_k(\mu, X,  X')$ of the resolvent and evolution operators 
 $\left(\widetilde{{\cal D}}_k+\lambda^2\right)^{-1}$, $\left(\Lambda^{\lambda, k}+\mu^2\right)^{-1}$, $e^{t \widetilde{{\cal D}}_k}$  and $e^{t \Lambda^{\lambda, k}}$.

The operators $\widetilde{{\cal D}}_k$ and $\Lambda^{\lambda, k}$ are
the modified Schr\"odinger operators with magnetic field  on the hyperbolic upper half plane $\H$  and the Schr\"odinger operator with the diatomic molecular Morse potential on the real line $\R$, they are respectively given by (Ikeda-Matsumoto \cite{IKEDA-MATSUMOTO})
\begin{align}\widetilde{{\cal D}}_{k}=y^2\left(\frac{\partial^2}{\partial x^2}+\frac{\partial^2}{\partial y^2}\right)+ 2 i k y\frac{\partial}{\partial x}+ \frac{1}{4},\end{align}
and
\begin{align}\label{Morse}\Lambda^{\lambda, k}=\frac{\partial^2}{\partial X^2}-2k\lambda e^X -\lambda^2 e^{2X}.\end{align}
Note that the operator $\widetilde{{\cal D}}_{k}$ is introduced firstly by Maass \cite{MAASS}  and can be written as
\begin{align}\widetilde{{\cal D}}_k= {\cal L}_{k}^{\H}+k^2+\frac{1}{4},\end{align} 
where ${\cal L}_{k}^{\H}$ is the  Schr\"odinger operator with uniform magnetic field on the hyperbolic upper half plane $\H$ 
\begin{align}{\cal L}_{k}^{\H}=-(z-\overline{z})^2)\frac{\partial^2}{\partial z\partial\overline{z}}+k(z-\overline{z})\left(\frac{\partial}{\partial z}+\frac{\partial
}{\partial \overline{z}}\right)-k^2.\end{align}


The operators  ${\cal L}_{k}^{\H}$   has a physical interpretation as being the Hamiltonian which governs a non relativistic charged particle
moving under the influence of the magnetic field of constant strength $|k|$, perpendicular to $\H$. The operator $\widetilde{{\cal D}}_{k}$  is basic to the work of Roelcke \cite{ROELCKE}, Elstrodt \cite{ELSTRODT}, Fay  \cite{FAY} and Comtet \cite{COMTET}.\\
 The purely vibrational levels of diatomic
molecules with angular momentum $l=0$ have been described by the Morse potential since $1929$ (see Morse\cite{MORSE}).\\
The operators  $\widetilde{{\cal D}}_{k}$ and $\Lambda^{\lambda, k}$ are,
non-positive, definite
and each of them has an absolute continuous spectrum and a points spectrum if $|k|\ge 1/2$ .\\
For a recent work on the magnetic field on the hyperbolic plane and the Morse Potential (see Ould Moustapha \cite{OULD MOUSTAPHA2} and the references therein).\\
 For $k=0$, the operator  $\widetilde{{\cal D}}_{0}$  reduces to the free Laplace-Beltrami operators on the hyperbolic half-plane and
the free resolvent kernel on the hyperbolic upper half plane 
is given by 
\begin{align*} \widetilde{G}_0(\mu, z, z')=
\frac{\Gamma(s)\Gamma(s)}{4\pi\Gamma(2s)}\cosh^{-2s}\rho( z,  z'){}_2F_1\left(s, s , 2s,  \cosh^{-2}\rho(z, z')\right),\end{align*}
where $s=(1-i\mu)/2$ and $\rho(z,  z')$ is  the hyperbolic half plane distance.\\
The function ${}_2F_1(a, b, c, z)$ is the Gauss hypergeometric function defined by:
\begin{align}\label{Gauss}
  {}_2F_1(a, b, c, z)=\sum_{n=0}^{\infty}\frac{(a)_n(b)_n}{(c)_n n!}z^n,
  \quad |z|<1,
\end{align}
  $(a)_n=\frac{\Gamma(a+n)}{\Gamma(a}$  is the Pochhamer symbol,
and $\Gamma$ is the classical Euler function.\\
For $k=0$ the operator $\Lambda^{\lambda, 0}$ reduces to the Schr\"odinger operator with the so-called Liouville potential. In this case the resolvent kernel is given by the explicit formula Ikeda-Matsumoto\cite{IKEDA-MATSUMOTO}
\begin{align}G^{M}_{\lambda, 0}(\mu, X, X')=2 I _{i\mu}(\lambda e^X)K_{i\mu}(\lambda e^{X'}).\end{align}
For the wave kernel of the operator $\Lambda^{\lambda, 0}$ see Abdelhaye et al. \cite{ABDELHAY et al.}\\
The organization of the paper is as follows. In sections 2 and 3 we give explicit formula for the resolvent kernel and we solve the time dependent  Schr\"odinger equation  with magnetic field 
on the hyperbolic half plane $\H$. Sections 4  and 5  are devoted to explicit formulas for the resolvent kernels and the solution of the time dependent  Schr\"odinger equation with diatomic molecular Morse
 potential on $\R$.
Finally in section 6 we introduce and discuss some applications.
\section{Resolvent kernel under the magnetic field on the hyperbolic half  plane $\H$ }
Let  $\H=\{z=x+iy\in \C,\, y>0\}$ be the upper half plane endowed with the usual hyperbolic metric $\widetilde{ds}=\frac{dx^2+dy^2}{y^2},$ 
the Riemannian space $(\H, \widetilde{ds})$ is called the hyperbolic upper half plane.
 The metric $\widetilde{ds}$ is invariant with respect to the group $\widetilde{G}=SL_2(\R)$ with
 $$ SL_2(\R)=\left\{ \left(
\begin{array}{cc}
a & b\\
c & d
\end{array}\right) a, b, c, d \in \R: a d-c b=1\right\}.$$
The hyperbolic distance between two points $z, z' $ is given by
\begin{align}\label{dist2} \cosh^2 (\rho(z, z')/2)=\frac{(x-x')^{2}+(y+y')^{2}}{4yy'}.\end{align} 
 It is well known that the hyperbolic upper half plane $(\H, \widetilde{ds})$ is isomorphic to the hyperbolic disc $(\D, ds)$ 
via the Cayley transform:
$$ z=c^{-1}w=-i\frac{w+1}{w-1},
w=c z=\frac{z-i}{z+i}. $$
Let
${\cal D}_{k}$ be  the modified Schr\"odinger operator with uniform magnetic field on the hyperbolic disc $\D$, defined by  (Ould Moustapha \cite{OULD MOUSTAPHA2})
 \begin{align*}{\cal D}_{k}=(1-|w|^2)^2\frac{\partial^2}{\partial w \partial \overline{w}}+k(1-|w|^2)w\frac{\partial}{\partial w} \nonumber\\
-k(1-|w|^2)\overline{w}\frac{\partial}{\partial \overline{w}}-k^2|w|^2+k^2+\frac{1}{4}. \end{align*}
The operators ${\cal D}_{k}$ and $\widetilde{{\cal D}}_{k}$ are related by the formula
\begin{align}\label{intertwin}U_{k}\widetilde{{\cal D}}_{k}U^{-1}_{k}f(c^{-1}w)={\cal D}_{k} f(c^{-1}w),\end{align}
with
\begin{align}\label{U}\left(U_{k}f\right)(w)=\left(\frac{1-\overline{w}}{1-w}\right)^{k}f(c^{-1}w),\ \ \left(U^{-1}_{k}g\right)(z)=\left(\frac{i-\overline{z}}{z+i}\right)^{k}g(c z).\end{align} 
\begin{prop}\label{Resolvent} The resolvent kernel associated to the modified Schr\"odinger operators with magnetic field $\widetilde{{\cal D}}_{k}$ on the hyperbolic half plane $\H,$
 is given by
$\widetilde{G}_k(s, z, z')=\frac{\Gamma(s-k)\Gamma(s+k)}{4\pi\Gamma(2s)}\left(\frac{z'-\overline{ z}}{z-\overline{z'}}\right)^{k}\left(\cosh^2\rho(z, z')/2\right)^{-s}\times $  \begin{align}F\left(s-|k|, s+|k|, 2s, \cosh^{-2}\rho(z, z')/2\right) 
\end{align} 
with $\cosh^2\rho(z, z')/2$ is given in \eqref{dist2}
 \end{prop}
 \begin{proof} Using the formula \eqref{U} 
we have 
\begin{align} \label{transport}
U_k^z [U_{-k}^{z'}[\widetilde{G}_k(t, z, z')](w')](w')=G_k(s, w, w')
\end{align}
where $G_k(s, w, w')$ is resolvent kernel with modified uniform magnetic potential on the hyperbolic disc given by (Ould Moustapha \cite{OULD MOUSTAPHA2})\\
$\label{G}G_k(s, w, w')=
\frac{\Gamma(s-k)\Gamma(s+k)}{4\pi\Gamma(2s)}\left(\frac{1-w\overline{w'}}{1-\overline{w}w'}\right)^{k}\times $ 
\begin{align}\left(\cosh^2 (d(w, w')/2\right)^{-s}F\left(s-|k|, s+|k|, 2s, \cosh^{-2}(d(w, w')/2\right),
\end{align}
 \begin{align}\label{dist1}\cosh^2 (d(w, w')/2)=\frac{|1-w\overline{w'}|^2}{(1-|w|^2)(1-|w'|^2)}, \end{align}
and from the formula \eqref{intertwin} 
we obtain the results of the proposition.
\end{proof}
 \begin{prop}  Let $\widetilde{G}_k(\mu, z, z')$ be the resolvent kernel with uniform magnetic field on the hyperbolic upper half plane, then  we have \\ 
\begin{align}\label{key}\widetilde{G}_k(\mu, z, z')=\int_\rho^{+\infty}\widetilde{W}_k(b,\rho)\frac{e^{-i\mu b}}{2i\mu} d b,  \end{align}
where\\
$\widetilde{W}_k(b, \rho)=\frac{1}{2\pi}\left(\frac{z'-\overline{z}}{z-\overline{z'}}\right)^{k}
(\cosh^2b/2-\cosh^2\rho/2)^{-1/2}\times $ \begin{align}\label{Kernel} F\left(|k|,-|k|, 1/2,1-\frac{\cosh^2b/2}{\cosh^2\rho/2}\right),\end{align}
and  $\rho(z, z')$  is the distance between two points $z$ and $z'$ given in  \eqref{dist1}.
\end{prop}
\begin{proof}
The proof of this proposition is a consequence of the formula \eqref{transport}  and the  Theorem 2.3 (Ould Moustapha \cite{OULD MOUSTAPHA3}).
\end{proof}

\begin{prop} Let $\widetilde{W}_k(b, \rho)$ be the kernel in \eqref{Kernel}, the following formulas hold\\
{\bf i)}
$\widetilde{W}_k(b, \rho)=\frac{1}{2\pi}\left(\frac{z'-\bar z}{z-\bar{z'}}\right)^{k}(\cosh^2(b/2) - 
\cosh^2(\rho(z, z')/2))_+^{-1/2}\times$\\
$F(2|k|,-2|k|;1/2,(1-\frac{\cosh (b/2)}{\cosh (\rho/2)}).$\\
{\bf ii)}
$\widetilde{W}_k(b, \rho)=\frac{1}{2\pi}\left(\frac{z'-\bar z}{z-\bar{z'}}\right)^{k}\left(\frac{\cosh^2b/2}{\cosh^2\rho/2}\right)^{|k|}(\cosh^2(b/2) - 
\cosh^2(\rho(z,z')/2))_+^{-1/2}\times$\\
$ F(-|k|,1/2-|k|;1/2,1-\frac{\cosh^2 (\rho/2)}{\cosh^2 (b/2)}).$\\
{\bf iii)} For $|k|$ integer or half integer we have:\\
$\widetilde{W}_k(b, \rho)=\frac{1}{2\pi}\left(\frac{z'-\bar z}{z-\bar{z'}}\right)^{k}(\cosh^2(b/2) - 
\cosh^2(\rho(z,z')/2))_+^{-1/2}T_{2|k|}\left(\frac{\cosh (b/2)}{\cosh (\rho(z,z')/2)}\right)$\\
where $T_{2|k|}(x)$ is the Chebichev polynomial of the first kind.\\
{\bf iv)}  For $|k|$ integer or half integer we have\\
$\widetilde{W}_k(b, \rho)=\frac{1}{2\pi}\left(\frac{z'-\bar z}{z-\bar{z'}}\right)^{k}\left(\cosh(\rho(z, z')/2)\right)^{-2|k|}\times$ \\  $\sum_{n=0}^{n=[|k|]}\frac{(-|k|)_n (1/2-|k|)_n}{(1/2)_n n!}\cosh(b/2)^{2|k|-2n}\left(\cosh^2b/2-\cosh^2\rho/2\right)^{n-1/2}.$\\
where $[|k|]$ is the entire part of $|k|$.
\end{prop}
\begin{proof}
Using the formula Magnus et al. \cite{MAGNUS et al.},$ p.50$\\
$F(a, b, (a+b+1)/2, z)=F(a/2, b/2, (a+b+1)/2, 4z-4z^2)$
we get i).\\ 
To see ii) we apply the following formula Magnus et al. \cite{MAGNUS et al.},$p. 47$\\
$F(a, b, c, z)=(1-z)^{-a}F(a, c-b, c, \frac{z}{z-1})$.\\
The part iii) is a consequence of the formula  Magnus et al. \cite{MAGNUS et al.}, $p.39$,\\
$T_n(1-2x)=F(-n, n, 1/2, x)$.\\
Finally iv) can be shown from ii) and  the formula  Magnus et al. \cite{MAGNUS et al.} p.$39$ 
$F(-m, b, c, z)=\sum_{n=0}^{n=[m]}\frac{(-m)_n (b)_n}{(c)_n n!}z^n$
and the Proof of Lemma is finished.\\
\end{proof}
\section{The time dependent Schr\"odinger equation under a magnetic field on the hyperbolic upper half plane}
Now, we give the solution of the following time dependent Schr\"odinger equation under a uniform magnetic field on the hyperbolic upper half plane. 
\begin{align}\label{Cauchy-Heat} \left \{\begin{array}{cc}\widetilde{\cal D}_k u(t, z)=\frac{\partial}{\partial
t}u(t, z)
, (t, z)\in \R^\ast_+\times \H \\ u(0, z)=0, u_t(0, z)=u_1(z), u_1\in
C^\infty_0(\H)\end{array}
\right., \end{align}
\begin{cor} The Schwartz integral  kernel of the heat operator $e^{t \widetilde{\cal D}_k}$  that solves the heat Cauchy problem\eqref{Cauchy-Heat}\ with uniform magnetic field on the hyperbolic plane is given by:
\begin{align}\label{heat kernel}H_k(t, z, z')=\int_r^\infty\frac{ e^{-b^2/4t}}{(4\pi t)^{3/2}}\widetilde{W}_k(b, z,  z')b db,\end{align}
where\\
$\widetilde{W}_k(r, \rho)=\frac{1}{2\pi}\left(\frac{z'-\overline{z}}{z-\overline{z'}}\right)^{k}
(\cosh^2r/2-\cosh^2\rho/2)^{-1/2}\times $ 
\begin{align} F\left(|k|,-|k|, 1/2, 1-\frac{\cosh^2r/2}{\cosh^2\rho/2}\right),\end{align}
and $\rho(z, z')$  is the distance between two points $z$ and $z'$ given \eqref{dist1}.
\end{cor}
\begin{proof}
The resolvent kernel and the heat kernel are related by the Laplace and Laplace inverse transforms as
\begin{align} G_k(\mu, w, w')=\int_0^\infty e^{\mu^2 t}H_k(t, w, w')dt,
\end{align}
or
\begin{align} H_k(t, z, z')=\frac{-1}{2 i\pi}\int_{c-i\infty}^{c+i\infty} e^{\mu t}G_k(\sqrt{\mu}, z, z')d\mu.
\end{align}
By combining the formula \eqref{key} and the formula  given by Prudnikov et al.  \cite{PRUDNIKOV et al.} p.52 
\begin{align}\label{Laplace-Transform} L^{-1}e^{-a\sqrt{p}}(x)=\frac{a}{\sqrt{4\pi x^3}}e^{-\frac{a^2}{4 x}},\ \ \  {\cal R}e p >0, \ \ \  {\cal R}e a^2 >0 ,\end{align}
and using Fubini theorem we get 
the formulas  \eqref{heat kernel} and the proof of the Corollary is finished.
\end{proof}
\section{Resolvent kernel with the  Morse potential on $\R$}
The Schr\"odinger operator with the Morse potential is given for $X\in \R$ by \eqref{Morse}
or equivalently for by setting $y=e^X \in \R^+$ by
\begin{align}\Lambda^{\lambda, k}_{\ln y}=\left(y\frac{\partial}{\partial y}\right)^{2}-2k\lambda y -\lambda^2 y^2.\end{align}
\begin{prop}
i) The modified Schr\"odinger operator with uniform magnetic field $\widetilde{{\cal D}}_{k}$ on the hyperbolic half plane  and  Schr\"odinger operator with the Morse diatomic molecular potential $\Lambda^{\lambda, k}_{\ln y}$ on the real line $\R$  are connected via the formulas
\begin{align}{\cal F}_x\left[y^{-1/2}{\cal D}^z_{k}y^{1/2}\Phi\right](\lambda, y)=\Lambda^{\lambda, k}_{\ln y}\left({\cal F}\Phi\right)(\lambda, y),\end{align}
where the Fourier transform is given by
\begin{align}
[{\cal F}f](\xi)=\frac{1}{\sqrt{2\pi}}\int_{\R}e^{-i x \xi}f(x)dx.
\end{align}
ii) The resolvent kernels  $\widetilde{G}_k(\mu, y, y')$ with uniform magnetic potential on the hyperbolic half plane and  $\widetilde{{\cal G}}_{\lambda,  k}(t, y, y')$  with the Morse potential are related by the formula
\begin{align}\label{connection}{\cal G}_{\lambda,  k}(\mu, y, y')= \frac{1}{\sqrt{y y'}}\int_{-\infty}^{\infty}e^{-i\lambda(x-x')}\widetilde{G}_k(\mu, z, z')d(x-x').\end{align}
\end{prop}
\begin{proof} The part i) is simple and in consequence is left to the reader. The part  ii) is a consequence of i).\\
\end{proof}
\begin{cor} The Schwartz kernel of the resolvent operator for the Schr\"odinger operator with Morse potential  
 is given by
\begin{align}\label{g} {\cal G}^{M}_{\lambda, k}(y, y', \mu)=\int^{\infty}_{0}e^{-i\mu b}W_{\lambda, k}(b, y, y') db,
\end{align}
where
\begin{align}\label{connection}W_{\lambda, k}(b, y, y') = \frac{1}{\sqrt{y y'}}\int_{-\infty}^{\infty}e^{-i\lambda(x-x')}\widetilde{W}_k(b, z, z')d(x-x').\end{align}
\end{cor}
\begin{proof}
This corollary is a consequence of the formulas \eqref{key} and \eqref{connection}.
\end{proof}

 Now, we shall prove the following theorem.
\begin{thm} The Schwartz kernel of the resolvent operator for the Schr\"odinger operator with Morse potential  
 is given by
\begin{align}\label{key-key} {\cal G}^{M}_{\lambda, k}(\mu, X, X')=\int^{\infty}_{0}e^{-\i\mu b}W_{\lambda, k}(b, y, y') db,
\end{align}
where
\begin{align}\label{w}W_{\lambda, k}(b, y, y')=C_{k}(4 y y')^{-|k|}\left(\frac{\partial}{\sinh (b /2) \partial b}\right)^{2|k|}
\nonumber\\ \frac{(2Z)^{4|k|}}{(Z+Y)^{2|k|}} e^{-i\lambda Z}\Phi_1(2|k|+1/2, 2|k|, 4|k|+1, 2i\lambda Z, \frac{2 Z}{Z+Y}),\end{align}
 $C_k=\frac{(-1)^{k}\Gamma(2|k|+1/2)}{2\Gamma(4|k|+1)\sqrt{\pi}}$,
 $Z=\sqrt{4y y'\cosh^2 (b/2)-(y+y')^2},$ \\ 
$Y=i\,sign(k)(y+y')$,
with the function $\Phi_1(a, b, c,  x, y)$ is the confluent hypergeometric function of two variables defined by the double series:(see for example 
Erdelyi et al.\cite{ERDELYI et al.} p. $225$).\\
\begin{align}\label{phi1}\Phi_1(a, b, c,  x, y)=\sum_{m, n \geq 0}\frac{(a)_{m+n}(b)_n}{c)_{m+n}m!n!}x^my^n,\end{align}
 for $|y|<1$ and its analytic continuation elsewhere.
\end{thm}
\begin{proof} Using the formula
$$(a^2+y^2)^{-1}=\int_0^\infty e^{- a x}\,\frac{\sin x y} {y}\, dx,$$
with $y=\sqrt{-\Lambda^{\lambda, k}}$ and $a=-i\lambda$, we can write
\begin{equation}\label{transmutation} (-\lambda^2-\Lambda^{\lambda, k})^{-1}=\int_0^\infty e^{- i\mu t}\,\frac{\sin t \sqrt{-\Lambda^{\lambda, k}}} {\sqrt{-\Lambda^{\lambda, k}}}\,dt,\end{equation}
comparing with the formulas \eqref{transmutation} and \eqref{key-key}, we have the result of the Corollary.
\end{proof}
\section{The time dependent Schr\"odinger equation under the Morse Potential on the real line $\R$} 
\begin{thm} For $k$ integer or half integer, the Cauchy problem  for the time dependent Schr\"odinger equation with
 the Morse Potential has the unique solution given by:
\begin{align} U(t, X)=\int_{-\infty}^\infty q_{\lambda, k}^M(t, X, X')f(X')dX',
\end{align}
where
\begin{align}\label{Heat-Kernel} q^M_{\lambda, k}(t, X, X')=\int_{|X-X'|}^<\infty> \frac{e^{-\frac{b^2}{4t} }}{(4\pi t)^{3/2}}W_{\lambda, k}^M (b, X, X')b db,\end{align}

$W_{\lambda, k}^M(b, X, X')=c_1(k)\left(\frac{d}{\sinh b/2 db}\right)^{2|k|}\\ \frac{Y^{4|k|}}{Z^{2|k|}}e^{(-k)(X+X')}e^{i\lambda z}\Phi_1(\alpha, \beta, \gamma, -2i \lambda Y, Z)$,\\
$c_1(k)=(-1)^{|k|+1}\frac{\Gamma(2|k|+1/2)}{2^{|k|}\Gamma(4|k|+1)\sqrt{\pi}},$\\
 $Y=2e^{(X+X')/2}\sqrt{\cosh^2 (b/2)-\cosh^2((X-X')/2)}$ and \\  $Z=\frac{2\sqrt{\cosh^2 (b/2)-\cosh^2((X-X')/2)}}{\sqrt{\cosh^2 (b/2)-\cosh^2((X-X')/2)}+i sign(k)\cosh((X-X')/2)},$\\
 $sign(k)=1$ if $k>0$ and $ sign(k)=-1$ if $k<0$.\end{thm}
\begin{proof}
The resolvent kernel and the heat kernel are related by the Laplace and Laplace inverse transforms as
\begin{align} {\cal G}^{M}_{\lambda, k}(X, X', \mu)=\int_0^\infty e^{\mu^2 t}q^{M}_{\lambda, k}(t, X, X') dt,
\end{align}
where
\begin{align} q^{M}_{\lambda, k}(t, X, X') =\frac{-1}{2 i\pi}\int_{c-i\infty}^{c+i\infty} e^{\mu t} {\cal G}^{M}_{\lambda, k}(X, X', \sqrt{\mu})d\mu.
\end{align}
By using the formulas \eqref{g}, \eqref{Laplace-Transform}, and the
 Fubuni theorem, we arrive at the formula \eqref{Heat-Kernel}.  The proof of the Corollary is finished.
\end{proof}

\section{Applications}
In this section, we introduce some applications.\\
{\bf  1- An integral representation of the product of two Whittaker functions.}\\
For $\alpha > 0$, $X> X' $, $\lambda k> 0$ the following formula holds
\begin{align}\label{Product}\frac{\Gamma(\alpha-k+1/2)}{\lambda\Gamma(1+2\alpha)}e^{-(x+y)/2}W_{k, \alpha}(2\lambda e^y)M_{k, \alpha}(2\lambda e^x)=\nonumber\\
\int^{\infty}_{0}e^{-\alpha b}W_{\lambda,k}^M(b, X, X')db\end{align}
where $W_{\lambda, k}$ is as in \eqref{w}.\\

Recall the the Schwartz kernel of resolvent of the Schr\"odinger operator with the Morse potential given by(Ikeda-Matsumoto \cite{IKEDA-MATSUMOTO} p. $83$)
\begin{align}\label{Resolvent-Morse}  G^{M}_{\lambda, k}(\mu, X, X')=\frac{\Gamma(\mu-|k|+1/2)}{\lambda\Gamma(1+2\mu)}e^{-(X+X')/2}W_{|k|, \mu}(2\lambda e^{X'})M_{|k|, \mu}(2\lambda e^X),\end{align}
where  $W_{k, \mu}(x)$ and $ M_{k, \mu}(x)$ are the Whittaker functions of the first and the second kind respectively.\\
Combining the formulas \eqref{Resolvent-Morse} and \eqref{key-key}, we obtain the formula \eqref{Product}.\\
{\bf 2 - An Integral formula} \\
Set\\
$ J(t, y, y')=\frac{\lambda\sqrt{2y y'}}{i(\pi^3 t)^{1/2}}\int^{\infty}_{0}\int^{\infty}_{0}\frac{\sinh \xi }{\sinh^2 u}\times $\\  $\exp{\left(-\lambda(y+y')\coth u -\frac{2\lambda\sqrt{y y'}\cosh\xi}{\sinh u}+2k u+2\frac{(\pi+i\xi)^2}{t}\right)} d\xi du$\\
For $\lambda > 0$, $k\in R$ and  $t>0$ , we  have \begin{align} \label{Integral-Formula}{\cal I} m [J(t, y, y')]=\int_{|X-X'|}^\infty \frac{e^{-\frac{b^2}{4t}}}{(4\pi t)^{3/2}}W_{\lambda, k}^M (b, y, y')b db.\end{align}
 where $W_{\lambda, k}$ is as in \eqref{w}.


To see the above formula  we recall the following formula for $q^M_{\lambda, k}(t, x, y)$ given in Alili et al.\cite{ALILI et al.}
\begin{align}\label{Morse-Heat-Alili} q^M_{\lambda, k}(t, x, y)=\int^{\infty}_{0}e^{2ku}\frac{1}{2\sinh u}\exp{\left(-\lambda(e^x+e^y)\coth u\right)}\theta_{\bar{\Phi}}(t/4)du,\end{align}
where $$\theta_r(t)=\frac{r}{(2\pi^3 t)^{1/2}}e^{\pi^2/2t}\int^{\infty}_{0}e^{-\xi^2/2t}e^{-r\cosh \xi}\sinh \xi \sin(\pi\xi/t)d\xi,$$
and
$$\bar{\Phi}=\frac{2\lambda\exp{(x+y)/2}}{\sinh u}.$$
In view of the uniqueness of the solution of the time dependent  Schr\"odinger equation with the Morse potential we obtain the formula\eqref{Integral-Formula} by comparing \eqref{Morse-Heat-Alili} and \eqref{Heat-Kernel}.
 Note that using formulas Magnus et al. \cite{MAGNUS et al.} p.$305$,
$$W_{0, \alpha}(z)=\sqrt{z/\pi}K_\alpha(z/2),\ \ \ 
M_{0, \alpha}(z)=2^{2\alpha}\Gamma(\alpha+1)\sqrt{z}I_\alpha(z/2),$$
where $I_\nu$ and $K_\nu$ are the Modified Bessel functions of the first and the second kind
and   (\cite{MAGNUS et al.} p. 283)
$$\Phi_1(\alpha, 0, \gamma, x, y)={}_1F_1(\alpha, \gamma, x), \ \  \
F(\nu+1/2, 2\nu+1, 2 i z)=\Gamma(1+\nu)e^{i z}(z/2)^{-\nu}J_{\nu}(z),$$
with ${}_1F_1$ and $J_{\nu}$ are respectively the confluent hypergeometric and the Bessel functions of the first kind, 
 we have for $k=0$, 
\begin{align}\label{0} {\cal G}^{M}_{\lambda,  0}(\mu, X, X')=\int^{\infty}_{0}e^{-i\mu b}W_{\lambda, 0}(b, y, y') db,
\end{align}
$$W_{\lambda, 0}(b, y, y')=\frac{1}{2}J_0(|\lambda|\sqrt{2 y y' \cosh b-y^2-y'^2}).$$ 
The formula \eqref{Product} reduces to the following integral representation of the product of the Bessel
functions
Lebdev\cite{LEBEDEV} p. 140
$$I_\mu(u)K_\mu(v)=(1/2)\int^{\infty}_{0}e^{-i\mu b}J_0(\sqrt{2 u v\cosh b-u^2-v^2})db, u >0, v>0.$$

  M.V. Ould Moustapha, \textsc{Department of Mathematics,
 College of Arts and Sciences-Gurayat,
 Jouf University-Kingdom of Saudi Arabia }.\\
\textsc{Facult\'e des Sciences et Techniques
Universit\'e de  Nouakchott Al-asriya,
Nouakchott-Mauritanie.}\\
  \textit{E-mail address}: \texttt{mohamedvall.ouldmoustapha230@gmail.com}
  \end{document}